\documentclass[lettersize,journal]{IEEEtran}
\usepackage{amsmath,amsfonts}
\usepackage{algorithmic}
\usepackage{algorithm}
\usepackage{array}
\usepackage[caption=false,font=normalsize,labelfont=sf,textfont=sf]{subfig}
\usepackage{textcomp}
\usepackage{stfloats}
\usepackage{url}
\usepackage{verbatim}
\usepackage{graphicx}
\usepackage{cite}
\usepackage{hyperref}
\usepackage[ruled,vlined,algo2e]{algorithm2e}
\usepackage{amsthm}
\newtheorem{theorem}{Theorem}
\newtheorem{corollary}{Corollary}[theorem]
\theoremstyle{definition}
\newtheorem{definition}{Definition}
\theoremstyle{remark}
\newtheorem*{remark}{Remark}

\begin{document}

\title{Reinforcement Learning Generalization for Nonlinear Systems Through Dual-Scale Homogeneity Transformations}

\author{Abdel Gafoor Haddad,~Igor Boiko,~\IEEEmembership{Senior Member,~IEEE,},~Yahya Zweiri,~\IEEEmembership{Member,~IEEE}
}

\markboth{Submitted to IEEE Transactions on Control Systems Technology, October~2023}%
{Shell \MakeLowercase{\textit{et al.}}: A Sample Article Using IEEEtran.cls for IEEE Journals}


\maketitle

\begin{abstract}
Reinforcement learning is an emerging approach to control dynamical systems for which classical approaches are difficult to apply. However, trained agents may not generalize against the variations of system parameters. This paper presents the concept of dual-scale homogeneity, an important property in understating the scaling behavior of nonlinear systems. Furthermore, it also presents an effective yet simple approach to designing a parameter-dependent control law that homogenizes a nonlinear system. The presented approach is applied to two systems, demonstrating its ability to provide a consistent performance irrespective of parameters variations. To demonstrate the practicality of the proposed approach, the control policy is generated by a deep deterministic policy gradient to control the load position of a quadrotor with a slung load. The proposed synergy between the homogeneity transformations and reinforcement learning yields superior performance compared to other recent learning-based control techniques. It achieves a success rate of 96\% in bringing the load to its designated target with a 3D RMSE of $\mathbf{0.0253}$ m. The video that shows the experimental results along with a summary of the paper is available at \href{https://youtu.be/3VtwJI-p_T8}{this link}.
\end{abstract}


\section{Introduction}
\label{sec:introduction}
Reinforcement Learning (RL) has evolved as a potential alternative to classical control techniques in controlling dynamic systems under the assumption of having a Markov Decision Process (MDP). It is advantageous for systems that are highly nonlinear, as suggested in~\cite{auto}, and stochastic since it performs the learning in an elegant trial-and-error process as presented by~\cite{sutton}. It found many applications in fixed~\cite{ifacc} and mobile robotics in general~\cite{isa3}, multirotors~\cite{isa1}, and collaborative robotics such as unmanned surface vehicles and multirotors for search and rescue applications~\cite{isa2}.

An open challenge in applying RL to robotics is its generalization to systems that have different values of the parameters than the one used in training. This problem is commonly addressed by training on a high-fidelity simulation model with different parameters and/or improving the policy training process. In the latter category, domain randomization (DR), proposed by~\cite{random1}, is a popular direction of research where a strong variability is induced in the simulator that leads to a model that can deal with such variations. While it has been argued that the benefit of the model improvement beyond a certain point is negligible, both directions impose a huge computational burden. 

When the RL agent is trained in a simulator, it often does not transfer effectively to the real world. This issue is termed sim2real gap. Bridging the sim2real gap through DR has recently gained momentum in robotics applications. The typical element of these techniques is the change of the parameters in the physics simulator, including the system dynamics~\cite{random2}. By this, the trained model gets exposed to a wide range of possibilities of the environment as presented in~\cite{random1}. This resulted in applications where the machine learning model can be trained entirely in simulations and transferred successfully to the physical world without further training in a plethora of supervised~\cite{random1,random4,random6,random9} and reinforcement learning methods~\cite{random2,random3,random5}.

The authors in~\cite{random3} overcame the sim2real gap for flying a drone by randomizing the visual data in the environment. Their policy, composed of a deep neural network, maps the observed monocular images directly to 3D drone velocities. However, dealing with high-precision tasks was lacking in their demonstrated experiments which focused on collision avoidance in open spaces. Although~\cite{random3} used synthetic images, the results are extended by DR~\cite{random1} where it's applied to precision tasks, and the entire training is performed using non-realistic texture, unlike most object detection results in the literature. In~\cite{random8}, it was shown that sampling parameters from uniform distributions in DR may lead to suboptimal policies with high variance. The active DR that they propose seeks the most informative variations in the environment and follows this sampling strategy to train more frequently on these instances. The work of~\cite{random10} employed Bayesian optimization to yield a data-efficient DR. It was found that varying the kinematic parameters during the training can outperform the randomization of the dynamic ones, as presented by~\cite{random11}.

Although plenty of DR methods have been developed, their computational burden is an obstacle for practical implementations. Moreover, the DR strategy does not guarantee the performance under large parameter variations. While it could work on systems with different parameters, a further change beyond a certain range can result in the policy failure as we demonstrate in this paper. Moreover, the DR approach may address only parametric differences between the model used in training and actual dynamics, but not structural differences such as existence of additional dynamics. An effective method to guarantee the performance of a controller is to exploit the homogeneity property of the system~\cite{homobook}. Homogeneity in our work is understood as invariance of certain properties of a system to transformations involving parameters of the system. Therefore, the homogeneity property requires a base controller that can achieve the required control objective for a plant having certain "nominal" parameter values. Designing a base controller for complex nonlinear systems using classical methods requires huge amounts of effort.

One of the benchmark problems that can be considered is the task of swinging up and stabilizing an inverted pendulum. Compared to classical nonlinear control designs such as energy control~\cite{energy}, RL can generate an optimal control policy with minimal knowledge about the system. It is computationally less expensive as compared to model predictive control~\cite{mpc}. The benchmark problem can be used mainly to compare the proposed method to similar learning approaches such as the ones based on DR. The RL methods that employ DR exhibits high robustness to uncertainties in the parameters. For this, we use Recursive Least Squares (RLS) in an adaptation loop to estimate the parameters of the system. The robust feature that comes from the adaptive scheme is combined with the optimality of the base controller that is obtained by RL and its homogeneous form.

Furthermore, the control of a UAV with a slung load system is considered as a more complex robotics application~\cite{basheer} for our method. Load transportation using UAVs allows reaching areas that clearance distance are two valuable traits of this transportation mechanism. Recent works provide solutions to this transportation problem through trajectory planning and control. In~\cite{traj_plan}, the objective is to minimize the travel time of the UAV through an optimal trajectory planning paradigm. In~\cite{diff_flat2}, the UAV with slung load is given by a model for which differential flatness can be applied and a nonlinear control is designed. However, parametric uncertainties are not considered and the UAV/load masses are kept constant. The control in~\cite{min_swing1} is an antisway tracking control based on the linear quadrotor control theory. In contrast to our work which makes use of the load swing, antisway controllers~\cite{min_swing1,min_swing2} minimize the swing of the suspended load during transportation.

In~\cite{rel1}, the authors investigated how to control a quadrotor UAV carrying a hanging load. They suggested an energy-based nonlinear controller to manage the UAV's position and the swing angle of the payload. They also created an adaptive control scheme to account for the wire's unknown length with proven stability. However, the masses of the UAV and load were known and constant. Moreover, a fixed length was selected across all experiments. The work in~\cite{wind} considers uncertainties in the control design, but this uncertainty comes from the wind effect whereas the parameters of the system are kept constant and known. Tracking a load reference position has been the target of recent works~\cite{load_track}. Geometric control is also employed for the tracking problem for a load that is connected to multiple UAVs through rigid~\cite{tcst1} and elastic links~\cite{tcst2}. A similar setup is shown in~\cite{vel_track} where multiple UAVs sharing a suspended payload track a reference velocity through a cooperative control algorithm. However, our work focuses on the last step of the load delivery which is to bring the load to a certain position by exploiting the swing feature rather than tracking a reference trajectory.

In this work, we vary the length of the slung load cable to demonstrate that the homogeneity property can aid the RL-based nonlinear policies and lead to a performance that exceeds the capabilities of plain RL policies. Moreover, this performance remains invariant upon parameter variations.

Therefore, this paper addresses the classical control and DR shortcomings through combining ideas from homogeneous systems, adaptive control, and learning-based control, leading to the following contributions:
\begin{enumerate}
    \item Developing a homogeneous control approach that is insensitive to variations in the parameters of nonlinear systems.
    \item Combining a single reinforcement learning policy which is trained for one set of parameters with the developed homogeneity transformations to consistently attain optimal performance regardless of variations in system parameters.
    \item Conducting simulations and experiments to validate the developed approach on a quadrotor with a slung load in which a superior performance compared to the recent learning-based control techniques was achieved.
\end{enumerate}

The organization of the rest of the paper is as follows. In Section~\ref{prob_form}, the problem formulation is precisely given. The dual-scale homogeneity is defined and detailed. Section~\ref{approach} starts with the RL base control design and presents the homogeneity for the specific cases of the inverted pendulum and driver with slung load system. Section~\ref{results} shows extensive simulation results as well as real-world UAVs with slung load experimental results. Finally, Section~\ref{conclusion} concludes the work with some remarks on the achieved results, their importance, and the possible future extensions.

\section{Problem Formulation}
\label{prob_form}
\begin{definition}[Dual Scale Homogeneity]
Consider the nonlinear system modeled by the state equation
\begin{equation}
    \dot{\pmb x}=f(\pmb x,\pmb \alpha),~~~~\pmb x(0)=\pmb x_o
    \label{nonlinear}
\end{equation}
where $\pmb \alpha$ is the parameter vector that can, but not necessarily, vary in real-time. The states and parameters vectors are expressed as $\pmb x=[x_1~x_2~...~x_n]^T$ and $\pmb \alpha=[\alpha_1~\alpha_2~...~\alpha_{j}]^T$, respectively. If there exists a transformation $\tau=\zeta t$, $\pmb z=g(\pmb x,\pmb \alpha,\zeta)$, $\pmb z\in R^n$, and $\pmb \beta=h(\pmb \alpha,\zeta)$, $\pmb \beta\in R^k$, where $k<j$ such that the solution of the system
\begin{equation}
    \frac{d\pmb {z}}{d\tau}=f^*(\pmb z,\pmb \beta),~~~~\pmb z(0)=\pmb z_o=g(\pmb x_o,\pmb \alpha,\zeta)
    \label{nonlinear2}
\end{equation}
in scaled time $\tau$ satisfies the condition
\begin{equation}
    \pmb z(\tau)=\kappa \pmb x(\zeta t)
\label{eq_defi}
\end{equation}
then the systems~\eqref{nonlinear} and~\eqref{nonlinear2} are dual-scale\footnote{with respect to time and magnitude transformations.} homogeneous, and the property~\eqref{eq_defi} is the dual-scale homogeneity.
\end{definition}

Since
\begin{equation}
\begin{aligned}
    \frac{d\pmb {z}}{d\tau}&=\frac{\partial g}{\partial\pmb {x}}\frac{d\pmb {x}}{d\tau} + \frac{\partial g}{\partial \pmb {\alpha}}\frac{d\pmb {\alpha}}{d\tau} + \frac{\partial g}{\partial \zeta}\frac{d\zeta}{d\tau}\\
    &=\frac{\partial g}{\partial \pmb {x}}\dot{\pmb {x}}\frac{1}{\zeta}\\
    &=\frac{1}{\zeta}\frac{\partial g}{\partial\pmb {x}}f(\pmb {x},\pmb {\alpha}),
    \label{rhs}
\end{aligned}
\end{equation}
The problem is solved by finding the transformation $g(\pmb {x},\pmb {\alpha},\zeta)$ such that the right-hand side of~\eqref{rhs} has a number of parameters $k$ smaller than $j$.

In~\eqref{eq_defi}, $\kappa$ is a scalar and $\zeta$ is a non-negative scalar that represents the amplitude and time scaling of the system states. Despite the presence of structured invariance, this property is mathematically distinct from the traditional homogeneity property~\cite{homobook}.

The general spatio-temporal scale in~\eqref{eq_defi} has two special cases of interest:

(i) $\kappa=1$. In this case, the state has an identical amplitude profile but is temporally scaled.

(ii) $\zeta=1$. Here, the state has the same time profile but is spatially scaled.

\begin{remark}
Applying the spatio-temporal transformations given by $\tau=\zeta t$ and $\pmb z=g(\pmb x,\pmb \alpha,\zeta)$ to the system~\eqref{nonlinear} yields the transformed dynamics
\begin{equation}
    \frac{d\pmb{z}}{d\tau}=f(\frac{1}{\zeta}\pmb{z},\pmb{\alpha}).
    \label{sys4}
\end{equation}
That is, if one of the elements of $\pmb{x}$, such as $x_1$ is present in the function $f$ with coefficient $\alpha_1$, then selecting $\zeta=\alpha_1$ would result in the elimination of this coefficient in the equation of $z_1$. Thus, the system~\eqref{sys4} would have $k=j-1$ parameters present in the transformed dynamics.
\end{remark}

\begin{definition}[Homogenizing Control]
Let $\Bar u(.)$ in a system under control be a static map from system observations to control effort that is designed to achieve a specific control performance on a system that has nominal parameters lumped in the vector $\Bar{\pmb{\alpha}}$. Furthermore, let the nominal system have states $\Bar x$ in~\eqref{eq_defi} and the parameter-perturbed system be with states $x$. A dynamic system is dual-scale homogeneous with respect to its physical parameters if there exists a set of transformations to the control and feedback signals
\begin{equation}
    u=c_0\Bar{u}(c_1x_1,c_2x_2,...c_nx_n)
    \label{u_scaled}
\end{equation}
for which~\eqref{eq_defi} holds given that the initial states are scaled similarly as
\begin{equation}
    \pmb z(0)=\kappa [x_1(0)~~x_2(0)~~...~~x_n(0)]^T.
\end{equation}
The expression in~\eqref{u_scaled} is a homogenizing control.
\end{definition}
The problem boils down to finding the constants in~\eqref{u_scaled} which are typically functions of the system parameters $c_i=f(\pmb\alpha)$.

The homogeneity property is illustrated in Fig.~\ref{homog}.
\begin{figure*}[htbp]
\centerline{\includegraphics[width=14cm]{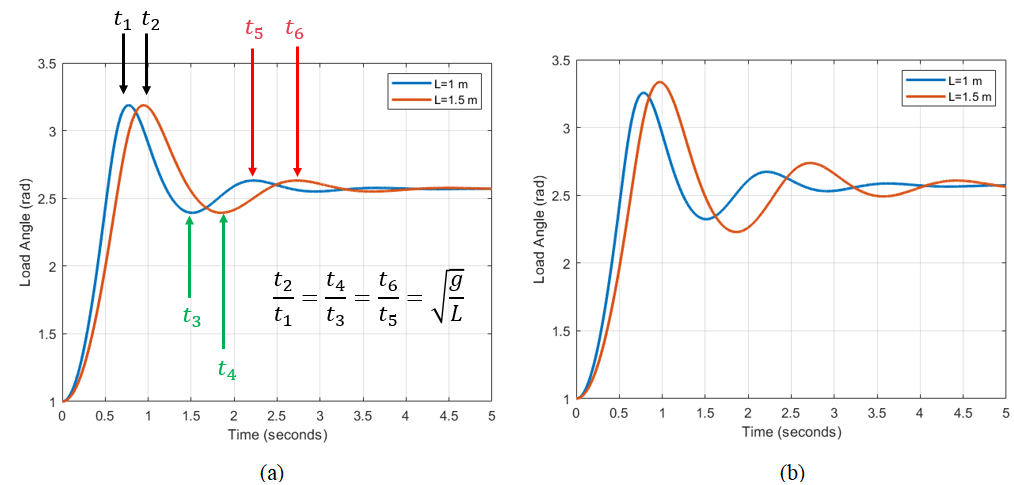}}
\caption{Illustration of the homogeneity property on the same plant: (a) homogeneity transformations are applied, hence, there exist $\kappa$ and $\zeta$ for which the transformation~\eqref{eq_defi} holds; (b) homogeneity transformations are not applied, thus,~\eqref{eq_defi} does not hold for any values of $\kappa$ and $\zeta$.}
\label{homog}
\end{figure*}

\begin{remark}
Consider the system
    \begin{equation}
        \begin{aligned}
    &\dot{\pmb {x}}=f(\pmb {x},\pmb {\alpha},u)\\
    &u=h(\pmb{x})
        \end{aligned}
    \end{equation}
to which the time scale $\tau=\zeta t$ will be applied. If the coordinate transformation $\pmb{z}=g(\pmb{x},\pmb{\alpha},\zeta)$ is a one-to-one map, the original coordinate can be expressed as
\begin{equation}
    \pmb{x}=g^{-1}(\pmb{z},\pmb{\alpha},\zeta).
\end{equation}
Therefore, the transformed system in the new coordinates $\pmb{z}$ and scaled time $\tau$ is given by
\begin{equation}
    \begin{aligned}
        &\frac{d\pmb{z}}{d\tau}=\frac{1}{\zeta}\frac{\partial g}{\partial\pmb{x}}f(g^{-1}(\pmb{z},\pmb{\alpha},\zeta),\pmb{\alpha},u)\\
        &u=h(g^{-1}\left(\pmb{z},\pmb{\alpha},\zeta)\right).
    \end{aligned}
\end{equation}
\end{remark}

This property allows the design of controllers that work for a wide range of parameters as will be demonstrated in the coming sections. It is especially beneficial when the control design is involved or computationally expensive, which is the case of modern reinforcement learning policies. In such conditions, the performance of the obtained policy can be harnessed despite the change in the system parameters without the need to retrain the agent.

\section{Approach}
\label{approach}
\subsection{Base Control Design}
The design of a nonlinear base control $u$ demonstrates the capability of the homogeneity transformations. For this, reinforcement learning (RL) is used. RL algorithms enable an agent to interact with an environment and learn through the data generated through that interaction in discrete time steps. The environment for the RL setup is considered an MDP. A finite MDP is composed of the following: 1) the state space $s = \{s_1, s_2,..., s_n\}$ composed of $n$ states and for which the vector $s_t$ refers to the state of the system at time $t$; 2) the set of actions that the agent can take at each state $A(s)$, where $a_t\in A(s_t)$ refers to the action taken by the agent at time $t$; 3) the transition function $T(s, a, s')$ that maps states and actions pair ($s$, $a$) to a probability distribution over next states $s'$; and 4) reward function $R(s, a, s')$ that computes the reward when moving from state $s$ to $s'$ as action $a$ is executed, where $r_t$ refers to the reward that the agent receives at time $t$. The flow of these signals during the training is shown in the block diagram in Fig~\ref{train}.

\begin{figure}[htbp]
\centerline{\includegraphics[width=4cm]{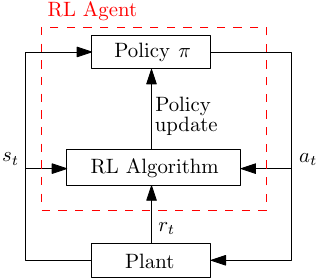}}
\caption{The flow of states, actions, and rewards during the training of RL agent. A nominal model of the plant is used for training in our method whereas the plant parameters are varied in DR.}
\label{train}
\end{figure}

In the setting of an RL problem, the agent aims to map the states to actions by learning, which leads to an optimal policy. Such policy is given as the mapping that yields the largest cumulative discounted rewards from all possible states in the next steps. The RL algorithm has to estimate the value function in order to learn the optimal policy. An agent that is in state $s$, takes an action $a$, and uses policy $\pi$ in the next steps has an expected cumulative discounted future reward $Q^\pi(s, a)$. Mathematically, the value function can be expressed as in \eqref{eqnQpolicy}.
\begin{equation}
\label{eqnQpolicy}
    Q^\pi(s,a)=\mathbb{E}\left(\sum_{i=0}^\infty{\gamma^ir_i|s_0=s, a_0=a,\pi}\right)
\end{equation}
where the discount factor $0\leq\gamma\leq1$ can be used to reduce the size of subsequent rewards relative to the immediate ones. Maximizing the average reward intake is another potential objective.

The agent can execute the action corresponding to the largest value in a state if the optimal $Q$-function $Q^*$ is known by $\pi^*(s)=arg max_a Q^*(s, a)$. In the experiments, we utilize the Deep Deterministic Policy Gradient (DDPG), an off-line model-free RL algorithm which we summarize in the next subsection.

The reinforcement learning agent is trained as per the DDPG algorithm by minimizing the loss~\cite{ddpg}
\begin{equation}
    L(\theta)=\mathbb{E}\left[(y-Q(s,a;\theta))^2\right]
    \label{loss}
\end{equation}
where the function approximator is parametrized by $\theta$.

Like other deep reinforcement learning methods, DDPG utilizes deep neural networks as value function approximator. It uses deterministic policy approximated by an actor neural network $\pi(s;\theta^\pi)$ that is a function of the $s$ and has weights $\theta^\pi$, and a critic network $Q(s,a;\theta^Q)$ which employs the Bellman equation to perform the update
\begin{equation}
    Q(s_t,a_t)=\mathbb{E}_{r_t,s_{t+1}}\left[r(s_t,a_t)+\gamma Q(s_{t+1},\pi(s_{t+1}))\right].
\end{equation}

On the other hand, the actor is updated through the chain rule applied to the loss function and updating $\theta^\pi$ based on the direction of the gradient of the loss~\eqref{loss} with respect to the weights
\begin{align*}
        \nabla_{\theta^\pi}L&\approx \mathbb{E}_s\left[\nabla_{\theta^\pi}Q(s,\pi(s|\theta^\pi)|\theta^Q)\right]\\
    &=\mathbb{E}_s\left[\nabla_{a}Q(s,a|\theta^Q)|_{a=\pi(s|\theta^\pi)}\nabla_{\theta^\pi}\pi(s|\theta^\pi)\right].
\end{align*}

As DDPG is an off-policy algorithm, the exploration behavior is encouraged by adding random noise $N$ to the policy, allowing the agent to explore areas in the action space. The stability of the training is enhanced by experience replay buffer $R$ and separate target networks. The latter can be done by having two additional networks, a target actor $\pi'$ and a target critic $Q'$ to compute the target $Q$ values. The update routine is included in the overall procedure in Algorithm~\ref{alg:two}. In order to avoid the temporal autocorrelation problem, the experience replay buffer stores the previous action-rewards and samples them to train the networks rather than using real-time data.

In this section, the mathematical models of the systems under consideration are used to obtain the corresponding homogeneity transformations. This step can be added at the beginning of ensemble methods to significantly improve their efficiency by reducing the dimension of the parameter space.

\subsection{Homogenization of Fixed Pendulum}
The nonlinear model of the inverted pendulum with negligible inertia and damping is given by
\begin{subequations}\label{eqab}
\begin{equation}
    \dot{x}_1 = x_2
\end{equation}
\begin{equation}
    \dot{x}_2 = \frac{g}{l}sin(x_1)+u\frac{1}{ml^2}
\end{equation}
\end{subequations}
where $x_1 = \Theta$ and $x_2 = \dot{\Theta}$. The pendulum length and mass are denoted by $m$ and $l$, respectively. The torque $u$ is applied directly to the pendulum pivot and $g$ is the gravitational acceleration constant.

Definition: The trained RL agent at $m=1$ kg and $l=g$ m is $\pi_n(z_1,z_2)$ which takes as arguments a transformed set of states and gives a normalized torque.

\begin{theorem}\label{th1}
For an arbitrary length $l$ and mass $m$, the control effort $u$ given by
\begin{equation}
    u=\frac{ml}{g}\pi_n\left(x_1,\sqrt{\frac{l}{g}}x_2\right)
    \label{homo_u}
\end{equation}
guarantees the property~\eqref{eq_defi} with both its conditions (i) and (ii) of the system~\eqref{eqab} for any perturbation in the value of $m$, $\dot{x}=y_1(t)$, and guarantees~\eqref{eq_defi} with condition (i) for any perturbation in the value of $l$.
\end{theorem}

\begin{proof}
The homogenization is performed as follows. Introduce a scaled time
\begin{equation}
    \tau=\zeta t
\end{equation}
where the time is accelerated when $\zeta>1$ and decelerated for $\zeta<1$. With this, the time derivatives of the states become
\begin{equation}
    \dot{x}_1=\frac{dx_1}{dt}=\zeta \frac{dx_1}{d\tau}
\end{equation}
\begin{equation}
    \dot{x}_2=\zeta \frac{dx_2}{d\tau}.
\end{equation}
The state equations are thus given by
\begin{subequations}\label{scaledab}
\begin{equation}
    \zeta \frac{dx_1}{d\tau}=x_2
\end{equation}
\begin{equation}
    \zeta \frac{l}{g} \frac{dx_2}{d\tau}=sin(x_1)+u\frac{1}{mgl}.
\end{equation}
\end{subequations}

Introduce the new state variables
\begin{equation}
    z_1 = x_1
\end{equation}
\begin{equation}
    z_2 = \zeta \frac{l}{g}x_2=\frac{\zeta^2l}{g}\frac{dx_1}{d\tau}.
\end{equation}
If we select $\zeta^2=g/l$, then we have
\begin{equation}
    z_2=\frac{dx_1}{d\tau}
\end{equation}
which yields
\begin{equation}
    \frac{dz_2}{d\tau}=\zeta\frac{l}{g}\frac{dx_2}{d\tau}.
\end{equation}
Substitute into the time-scaled state equations~\eqref{scaledab}
\begin{equation}
    \frac{dz_1}{d\tau}=z_2
\end{equation}
\begin{equation}\label{finaleq}
    \frac{dz_2}{d\tau}=sin(z_1)+f(z_1,z_2).
\end{equation}
where $f(z_1,z_2)=u(z_1,z_2)/mgl$. Since $u(z_1,z_2)$ can be selected such that~\eqref{finaleq} is not a function of any of the system parameters, the response of~\eqref{scaledab} is invariant to the change in $m$ and $l$, and the response of the original system~\eqref{eqab} is scaled by the time scaling used in the transformation.
\end{proof}
\begin{corollary}
The training of the RL agent is only needed at one point in the parameter space: $m=1$ kg and $l=g$ m. This is true since the policy is a static map given by the structure of a deep neural network with rectified linear unit (ReLU) and hyperbolic tangent activation functions.
\end{corollary}
\begin{corollary}
The time-scaled integral time-weighted absolute error (TS-ITAE) is invariant with respect to the physical parameters of the system. It is expressed as
\begin{equation}
    Q_{TS-ITAE}=\int_0^\infty{|z_i|\tau~d\tau}=Q_c=const
    \label{itae}
\end{equation}
where $Q_c$ is a constant cost which generally varies when the policy is modified but is invariant with respect to the changes in the system parameters.
\label{cor_itae}
\end{corollary}
Corollary~\ref{cor_itae} would mean that if the produced policy is optimal in the sense of~\eqref{itae} then this policy will also be optimal for other system parameters, once the derived homogeneity transformations are applied.

The control in~\eqref{homo_u} has the normalized policy $\pi_n$ as a function of the states and the system parameters. A prior knowledge of the parameters $\psi=[l,m]^T$ would allow a direct execution of the control law. However, an identification algorithm is required in case of missing or inaccurate values of the parameters. Furthermore, for time-varying parameters and to avoid interruption of operation, it is advantageous to have an algorithm that works in real-time. For this purpose, we apply RLS with forgetting factor as given by~\cite{sysid}
\begin{equation}
    \hat{\psi}(t)=\hat{\psi}(t-1)+L(t)\left[y(t)-\phi^T(t)\hat{\psi}(t-1)\right]
    \label{rls}
\end{equation}
\begin{equation}
    L(t)=\frac{P(t-1)\phi(t)}{\lambda(t)+\phi^T(t)P(t-1)\phi(t)}
\end{equation}
\begin{equation}
    P(t)=\frac{1}{\lambda(t)}\left[P(t-1)-\frac{P(t-1)\phi(t)\phi^T(t)P(t-1)}{\lambda(t)+\phi^T(t)P(t-1)\phi(t)}\right]
\end{equation}
where $\hat{\psi}$ is the estimated parameters vector, $y(t)$ is the measured output, $\phi$ is the regressors vector, $\lambda$ is the forgetting factor, and $P$ is the parameters covariance matrix. The overall procedure is summarized in Algorithm~\ref{alg:two}. 

\begin{algorithm2e}
\caption{Homogeneous Policy Based on Deep Deterministic Policy Gradient}\label{alg:two}
\KwData{Simulation environment: a frictionless pendulum model}
\KwResult{Homogeneous DDPG policy that swings up and stabilizes the pendulum with minimal effort and for any physical parameters}
Model transformation\\
Set parameters to the normalized values\\
Randomly initialize weights of $Q$ and $\pi$\\
Initialize target networks weights $\theta^{Q'}\gets \theta^Q$, $\theta^{\pi'}\gets \theta^\pi$\\
Initialize replay buffer $R$\\
\For{each episode}{
Initialize random process $N$ for action exploration\\
Receive initial observation state

\For{each step of episode}{
Execute DDPG training algorithm
}
}

Estimate the system parameters $\hat{\psi}$ by~\eqref{rls}\\
Transform $\pi$ to homogeneous control $u$ by~\eqref{homo_u}
\end{algorithm2e}

The execution of the policy with homogeneous transformation is shown in Fig.~\ref{homo}. The inputs to the RLS algorithm, which contains preprocessing and post-processing operations, are the true output $y(t)$ and the regressors $\phi(t)$. The latter is composed of the control $\phi_1(t)=u(t)$ which is the torque in this case, and $\phi_2(t)=sin(\Theta)$. The output $y(t)=\Ddot{\Theta}$ is the approximate derivative of the angular velocity which is provided by an encoder. All signals are discretized using the zero-order hold. The post-processing involves simple mathematical manipulations to retrieve the physical parameters. The coefficient of the second regressor is $g/l$ from which the length is obtained. Consequently, the mass is deduced from the coefficient of the first regressor which is $1/(ml^2)$.

\begin{figure}[htbp]
\centerline{\includegraphics[width=7cm]{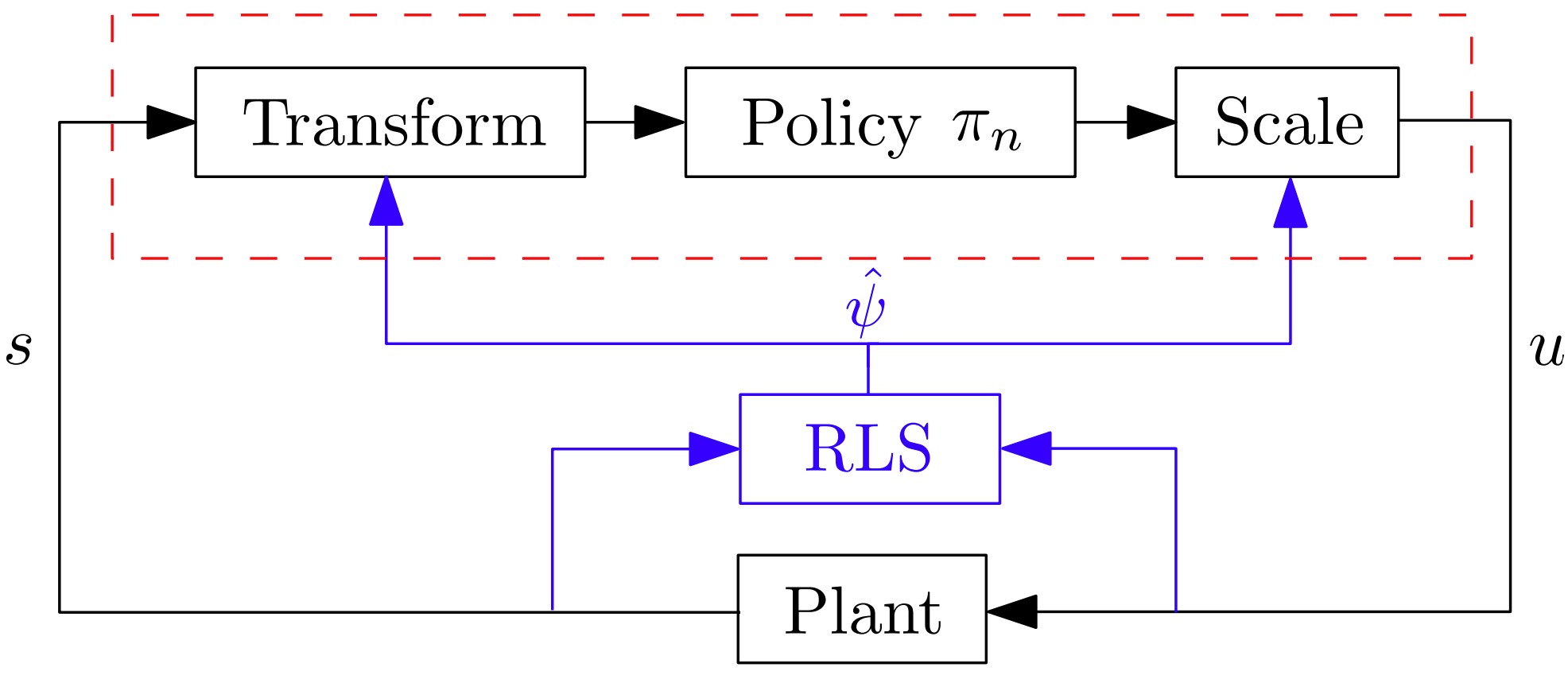}}
\caption{Block diagram of the system connection for the policy with homogeneous transformation. The adaptation loop is highlighted in blue. The state vector $s$ is composed of the two states $x_1$ and $x_2$.}
\label{homo}
\end{figure}

\subsection{Homogenization of Driver-Load}
Consider a model of a quadrotor moving in one dimension and carrying a load which makes an angle $\theta$. Neglecting the attitude dynamics, the model can be written as

\begin{equation}
\label{eq:qq1}
    (M+m)\Ddot{x}+mlsin(\theta)\dot{\theta}^2-mlcos(\theta)\Ddot{\theta}=F
\end{equation}
\begin{equation}
\label{eq:qq2}
        l\Ddot{\theta}-cos(\theta)\Ddot{x}+gsin(\theta)=0
\end{equation}
where $l$ is the length between the quadrotor and the load, $F$ is the applied force, and $M$ and $m$ are the masses of the quadrotor and load, respectively. Let $x_1=\theta$, $x_2=\dot{\theta}$, $x_3=x$, $x_4=\dot{x}$.

\begin{theorem}
For an arbitrary length $l$ and mass ratio $\mu=m/(M+m)$, the control effort $u$ given by
\begin{equation}
    u=mg\pi_n\left(x_1,\sqrt{\frac{l}{g}}x_2,\frac{1}{l}x_3,\frac{1}{\sqrt{lg}}x_4\right)
\end{equation}
guarantees the property~\eqref{eq_defi} with condition (i) in the state $\theta$ for any value of $l$, and~\eqref{eq_defi} in its general form for the position state.
\end{theorem}

\begin{proof}\label{proof1}
We can express the acceleration as
\begin{equation}
    \Ddot{x}=\mu \left(\frac{F}{m}+l\Ddot{\theta}cos\theta-l\dot{\theta}^2sin\theta\right)
\end{equation}
where $\mu={m}/({M+m})$. Substituting the result into the load equation yields
\begin{equation}
\label{eq:qq3}
    l\Ddot{\theta}-\mu cos\theta\left(\frac{F}{m}+l\Ddot{\theta}cos\theta-l\dot{\theta}^2sin\theta\right)+gsin\theta=0.
\end{equation}
We change the coordinates by introducing the variables
\begin{equation}\label{sub1}
    z_1=\theta;~~~z_2=\sqrt{\frac{l}{g}}\dot{\theta}
\end{equation}
from which the following relations can be deduced
\begin{equation}\label{tauqq}
    \dot{\theta}=\sqrt{\frac{g}{l}}z_2;~~~\dot{\theta}^2=\frac{g}{l}z_2^2;~~~\Ddot{\theta}=\frac{g}{l}\frac{dz_2}{d\tau}
\end{equation}
where the time scaling is $\tau=\sqrt{g/l}~t$. These relations can be used to rewrite~\eqref{eq:qq3} as
\begin{equation}
    g\frac{dz_2}{d\tau}-\mu cosz_1\left(\frac{F}{m}+g\frac{dz_2}{d\tau}cosz_1-gz_2^2sinz_1\right)+gsinz_1=0.
\end{equation}
Therefore, the following normalized dynamics are obtained
\begin{equation}
    \frac{dz_1}{d\tau}=z_2
\end{equation}
\begin{equation}
    \frac{dz_2}{d\tau}=\frac{\mu cosz_1\left(\frac{F}{gm}-z_2^2sinz_1\right)-sinz_1}{1-\mu cos^2z_1}.
\end{equation}
This guarantees the performance for $\theta$ similar to the proof of Theorem~1.

As for the position $x$, we express the angular acceleration as
\begin{equation}\label{subs11}
    \Ddot{\theta}=\frac{1}{l}\left(cos(\theta)\ddot{x}-gsin(\theta)\right).
\end{equation}
Substituting~\eqref{subs11} in~\eqref{eq:qq1} yields
\begin{equation}\label{bigqq}
(M+m)\ddot{x}+mlsin(\theta)\dot{\theta}^2-mcos(\theta)\left(cos(\theta)\ddot{x}-gsin(\theta)\right)=F.
\end{equation}
The same time scaling $\tau$ in~\eqref{tauqq} must be used for consistency. The following coordinate transformations are applied
\begin{equation}\label{sub2}
    z_3=\frac{1}{l}x;~~~z_4=\frac{dz_3}{d\tau}
\end{equation}
from which
\begin{equation}\label{sub3}
    z_4=\frac{1}{\sqrt{lg}}\dot{x};~~~\frac{dz_4}{d\tau}=\frac{1}{g}\ddot{x}.
\end{equation}
Substituting~\eqref{sub1},\eqref{tauqq},\eqref{sub2},\eqref{sub3} in~\eqref{bigqq} yields
\begin{equation}
    \frac{dz_4}{d\tau}=\frac{\frac{F}{mg}-sin(z_1)(z_2^2+cos(z_1))}{\mu^{-1}-cos^2(z_1)}
\end{equation}
from which the proof follows using the argument used in the proof of Theorem~1.
\end{proof}
Hence, the controller, such as the RL agent, needs to be trained only for certain values of $\mu=m/(M+m)$.

\subsection{Algorithmic Homogenization}
One of the challenges of the proposed approach is the required knowledge of the dynamics of each nonlinear system. This suggests developing a tool for the automation of the homogenization process. We present the automatic homogenizer which we make available to the public on Mathworks. It is important to note that the automated approach is also applicable to systems for which an accurate dynamic model is not available.

Genetic Algorithms (GA) can effectively solve many optimization problems including the identification of homogeneity transformations. However, the representation of individuals which is based on a fixed-length string-type is unnatural and constrains the effectiveness of GA. For instance, a function has more freedom when it is expressed as a hierarchical, variable-size structure instead of a fixed-length string with variable parameters. The initial choice of the string length in GA limits in advance the complexity of the function and sets an upper bound on what it can learn.

Genetic Programming (GP)~\cite{koza} can be used to learn and enhance nonlinear maps such as control laws~\cite{haddadgp}. This mapping can be represented by a recursive function tree, where the generations are obtained using the same operations used in genetic algorithms. First, the model is placed in a feedback loop using a controller with static gains. A state feedback controller is a simple arbitrary controller to generate the response data. Arbitrary gains are used since performance is not of concern in this identification routine. Two copies of the model are used. One is for the nominal model and another for the model with a variation of the parameter under investigation.

 A homogeneity cost function $J_h=min_{\gamma_h} |y_1(t)-y_2(\gamma_h t)|$ is defined based on the response characteristics of the system. This cost is fed back to the GP routine to guide it to a better transformation structure and/or better parameters selection. This continuously modified transformations takes the sensed states and supplies the scaled states to the nominal controller, and this process keeps repeating until some termination criteria is met. The homogeneity transformations generated by GP lead to identical responses to those generated analytically.

\section{Results and Discussions}
\label{results}
\subsection{Model Training}
The neural networks' framework is in line with the requirements of the DDPG approach discussed in Section 2.2. It consists of deep actor and critic networks, denoted by $\pi$ and $Q$, respectively. The long-term reward is approximated through the critic value function representation based on the observations and actions. Thus, the critic is a deep neural network that is composed of two branches. The first branch takes the observations through a fully connected layer of $400$ neurons, a rectified linear unit (ReLU) activation function, and another dense layer of $300$ neurons. The second branch takes the actions and processes the actions through one dense layer of $300$ neurons. The outputs of the two branches are summed and passed to a ReLU activation function to yield the critic output.

The actor network takes the current observations of the simulated system and decides the actions to take. This network starts with a fully connected layer of $400$ neurons, ReLU activation, another dense layer of $300$ neurons, and finally a hyperbolic tangent activation followed by an up-scale of the output to the torque limits. The learning rate of the actor is selected as $1\times10^{-4}$ whereas a larger learning rate $1\times10^{-3}$ is set for the critic. The discount factor of $0.99$ and noise variance of $0.6$ are used. The simulation time is set to $20$ s and the sampling time of the agent is $0.05$ s, leading to a $400$ steps per simulation. These values of hyper-parameters were obtained by tuning them to achieve the required swing up task.

The reward at every time step is given by
\begin{equation}
    r_t=-(\Theta_t^2+0.1\dot{\Theta}_t^2+0.0001u_{t-1}^2).
\end{equation}
Through this reward, the reinforcement learning algorithm receives feedback to know about the goodness of the actions that have been executed. The angle of the pendulum with respect to the vertical upright position is the primary component that forces the control to bring the pendulum to the desired position. As for the second term, if the pendulum starts moving rapidly, it could be harder to stabilize and more prone to instability. By promoting a smaller angular velocity, the agent is encouraged to take actions that prevent the pendulum from gaining unnecessarily high speeds in any direction. This was observed to aid in the convergence speed of the learning process. Finally, the last term ensures that the controller is efficient in terms of energy consumption and that a feasible control signal is applied. In a realistic scenario, the measurements of the angle and the angular rate are accessible through encoders.

The parameters of the pendulum are set to $m=1$~kg and $l=9.81$ m (value of $g$) as in Section 2.3. Furthermore, the input torque is saturated to $[-30,30]$ Nm during the training such that the pendulum cannot swing up without oscillating. The simulation environment is in MATLAB/Simulink with variable sampling time by Dormand-Prince solver. The average episode reward converged to the maximum at around 700 episodes, and further exploration was not required.

\subsection{Simulation Results}
In this subsection, the simulation results corresponding to the two presented theorems are provided. The code to reproduce the results is open-sourced on Matlab Central\footnote{https://www.mathworks.com/matlabcentral/fileexchange/154782-reinforcement-learning-with-homogeneity-transformations}.
\subsubsection{Inverted Pendulum}
This section shows the simulations of the inverted pendulum. The pendulum starts at the downward position which corresponds to $\pi$ rad and is swung up to the balanced position of $0$ rad. Also, the angle is wrapped to the interval $[-\pi,\pi]$. The action which is the torque, is continuous and its absolute value is saturated at $50$ Nm. The variation in the mass from $1$ kg to $2$ kg does not impact the performance as in Fig.~\ref{resp1}. The mass is not expected to change the performance due to the cancellation of the mass term in the equations of motion by the scaling in the homogeneous control $u$.
\begin{figure}[htbp]
\centerline{\includegraphics[width=9cm]{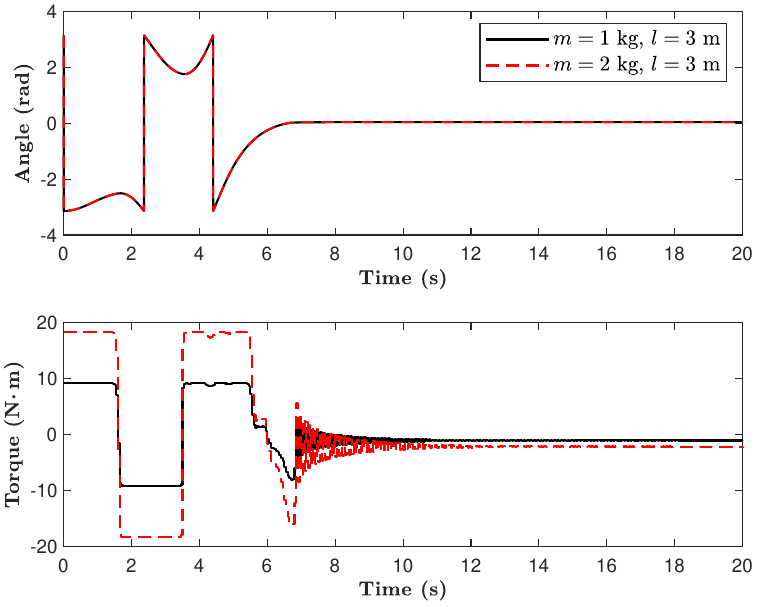}}
\caption{System response and control effort using the homogeneous RL policy with different mass values. The homogeneous policy resulted in a successful swing-up of the pendulum with exactly the same performance.}
\label{resp1}
\end{figure}
The change in the length results in slower or faster swing-up but without impact on the amplitude profile and the success of the swing-up as depicted in Fig.~\ref{resp2}.
\begin{figure}[H]
\centerline{\includegraphics[width=9cm]{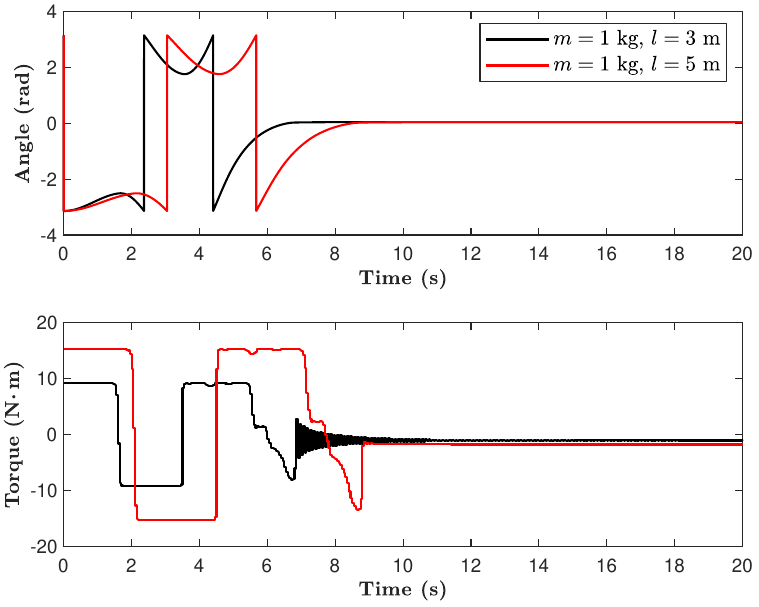}}
\caption{System response and control effort using the homogeneous RL policy where the length value is changed. The swing-up is successful with a different time scale.}
\label{resp2}
\end{figure}
The results also show a successful response of the RL agent trained with DR and another failed response. The training of the RL agent with DR was performed by varying the length and the mass independently over the intervals $[6,8]$ m and $[0.5,2.5]$ kg through a uniform distribution. Although the parameters are varied within the randomization interval used during training, DR-based policy failed as shown in Fig.~\ref{resp3}. Note also that the control effort in the successful attempt exhibits high frequency oscillations in the DR approach. The figure also shows the application of the proposed homogeneous DDPG policy which is able to swing up and stabilize the pendulum for all cases of parameter variations with a smooth control effort.
\begin{figure}[htbp]
\centerline{\includegraphics[width=9cm]{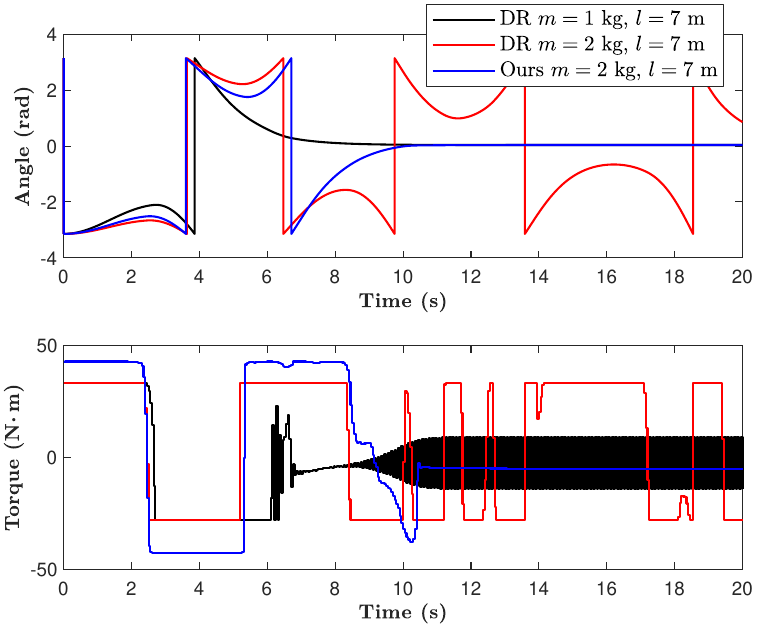}}
\caption{System angle response and control effort using the RL policy generated with DR and another using our proposed homogeneous DDPG controller.}
\label{resp3}
\end{figure}

The proposed method is tested also when the parameters of the system are unknown. Fig.~\ref{parameters} shows the estimations of the length and the mass along with their true values. The initial estimate of the length is $6$ m while the true value is initially $7$ m. Furthermore, the length changes to $8$ m in a step at $t=4$ s. For the mass, the initial estimate is $2.5$ kg whereas the true initial value is $2$ kg. The mass steps up to $3$ kg at $t=2$ s. The estimations are performed in real-time and converge to the true values while the pendulum is swinging up. A forgetting factor of $0.998$ is used to discount the older data and facilitate the convergence.
\begin{figure}[htbp]
\centerline{\includegraphics[width=9cm]{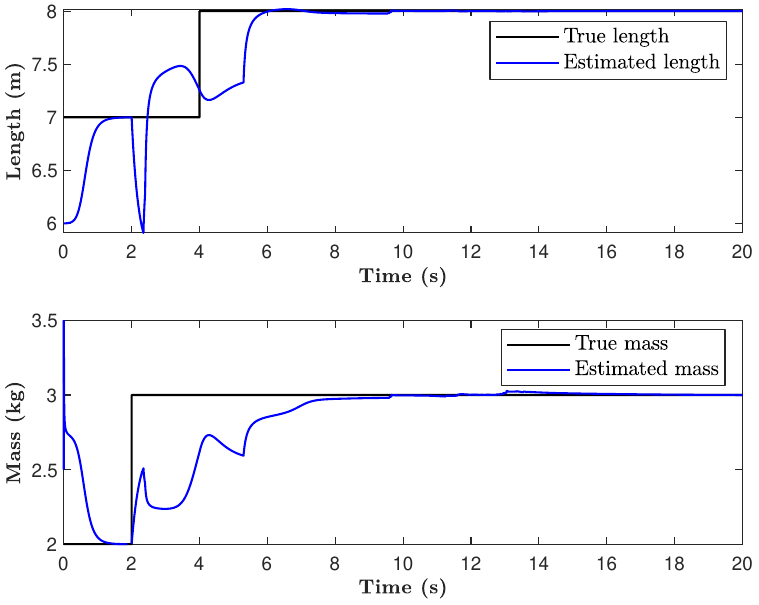}}
\caption{True values of the length and mass along with their online estimations using RLS. The value of the forgetting factor was selected empirically to strike a balance between the speed of convergence and estimation accuracy.}
\label{parameters}
\end{figure}

The response of the pendulum and the applied torque are shown in Fig.~\ref{resp_var}. It can be seen that the quick estimation of the parameters using RLS leads to a successful swing up. The robustness inherited by the RLS adaptation loop allows the proposed method to be used with systems with uncertainties in their parameters.
\begin{figure}[htbp]
\centerline{\includegraphics[width=9cm]{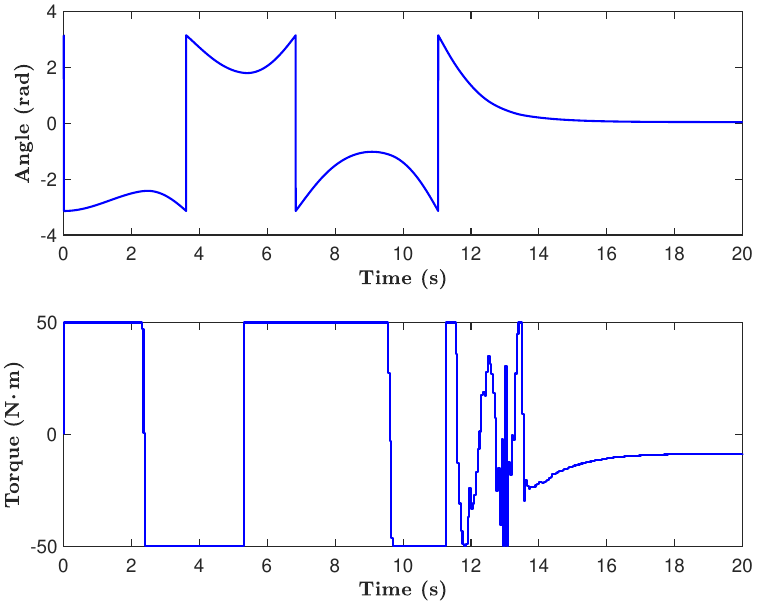}}
\caption{System response and control torque under the real-time variation of the length and mass.}
\label{resp_var}
\end{figure}

\subsubsection{Driver with Load}
This section shows the simulation results of the driver with load using the RL policy and the homogeneity transformations derived in Theorem 2. The position of the UAV and the load angle are shown in Fig.~\ref{sim_pos} and Fig.~\ref{sim_ang}, respectively. The cable length is 1 m in the nominal model whereas a length of 2.5 m is used in the varied case. It can be seen that the homogeneity property is maintained in both figures.
\begin{figure}[htbp]
\centerline{\includegraphics[width=9cm]{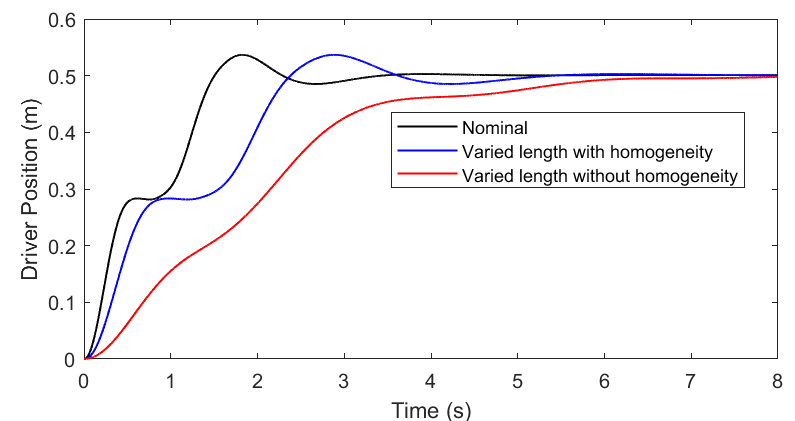}}
\caption{Driver-load position response in simulation. The response of the homogenized case is scaled to overlay the responses.}
\label{sim_pos}
\end{figure}

\begin{figure}[htbp]
\centerline{\includegraphics[width=9cm]{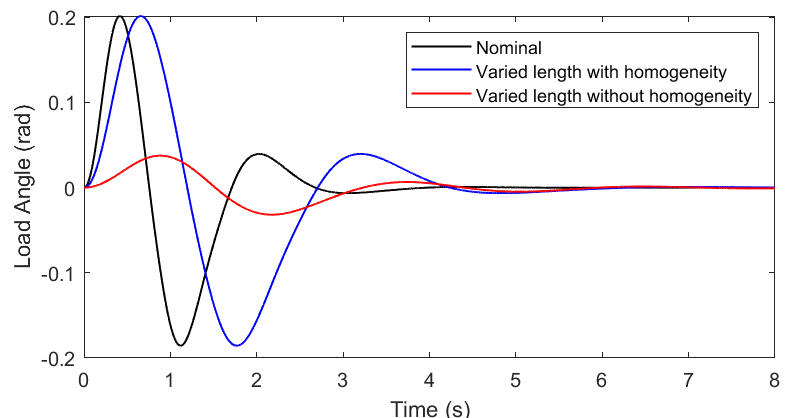}}
\caption{Driver-load angle response in simulation. The systems with homogeneity transformations exhibit identical amplitudes at different time scales.}
\label{sim_ang}
\end{figure}

\subsection{Experimental Results}
The real world experiments are performed on a quadrotor with slung load. We show the performance in bringing the load swing to a specific angle in comparison with the modern RL-based policies. The results can be seen as well in the accompanying video\footnote{\url{https://youtu.be/3VtwJI-p_T8}}. The block diagram of the control structure is shown in Fig.~\ref{sim_UAV}. 
\begin{figure}[htbp]
\centerline{\includegraphics[width=9cm]{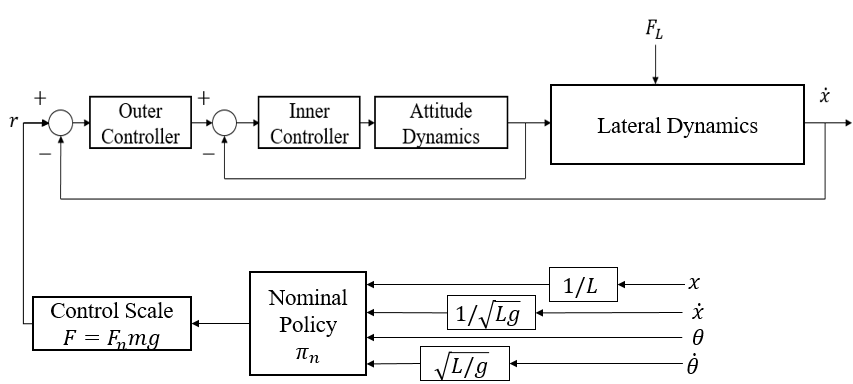}}
\caption{UAV with slung load system block diagram.}
\label{sim_UAV}
\end{figure}

The experimental setup is based on the Quanser QDrone. It weighs $M=850$ g with a thrust-to-weight ratio of 1.9 allowing it to carry $m_{max}=300$ g with a reasonable maneuvering ability. The on-board processing is performed on the Intel Aero Compute unit which is available off the shelf. The Madgwick
filter~\cite{mad} estimates the attitude of the QDrone using the BMI160 inertial measurement unit mounted on board. The position and yaw measurements of both the drone and the load is provided by the OptiTrack motion capture system at 250 Hz. Further details about the UAV parameters can be found in~\cite{dnn}. The angles that the load makes with respect to the drone are deduced from their relative positions. The payload is manufactured from mild steel to reduce its size and minimize its interaction with the propellers. The load weighs $m_b=62$ g and it is equipped with three markers for localization.

The experiments are performed using a nominal RL policy that has been trained to move the UAV in one direction by 0.75 m. The nominal case is a UAV with a slung load that is attached using a cable of length $L=0.5$ m as shown in Fig.~\ref{setup} whereas the cable length is $0.75$ m in the varied case. The homogeneity transformations are those derived in Theorem 2. The additional dynamics of the quadrotor were approximated by a time delay and the control signal was delayed to compensate for that.
\begin{figure}[htbp]
\centerline{\includegraphics[width=6cm]{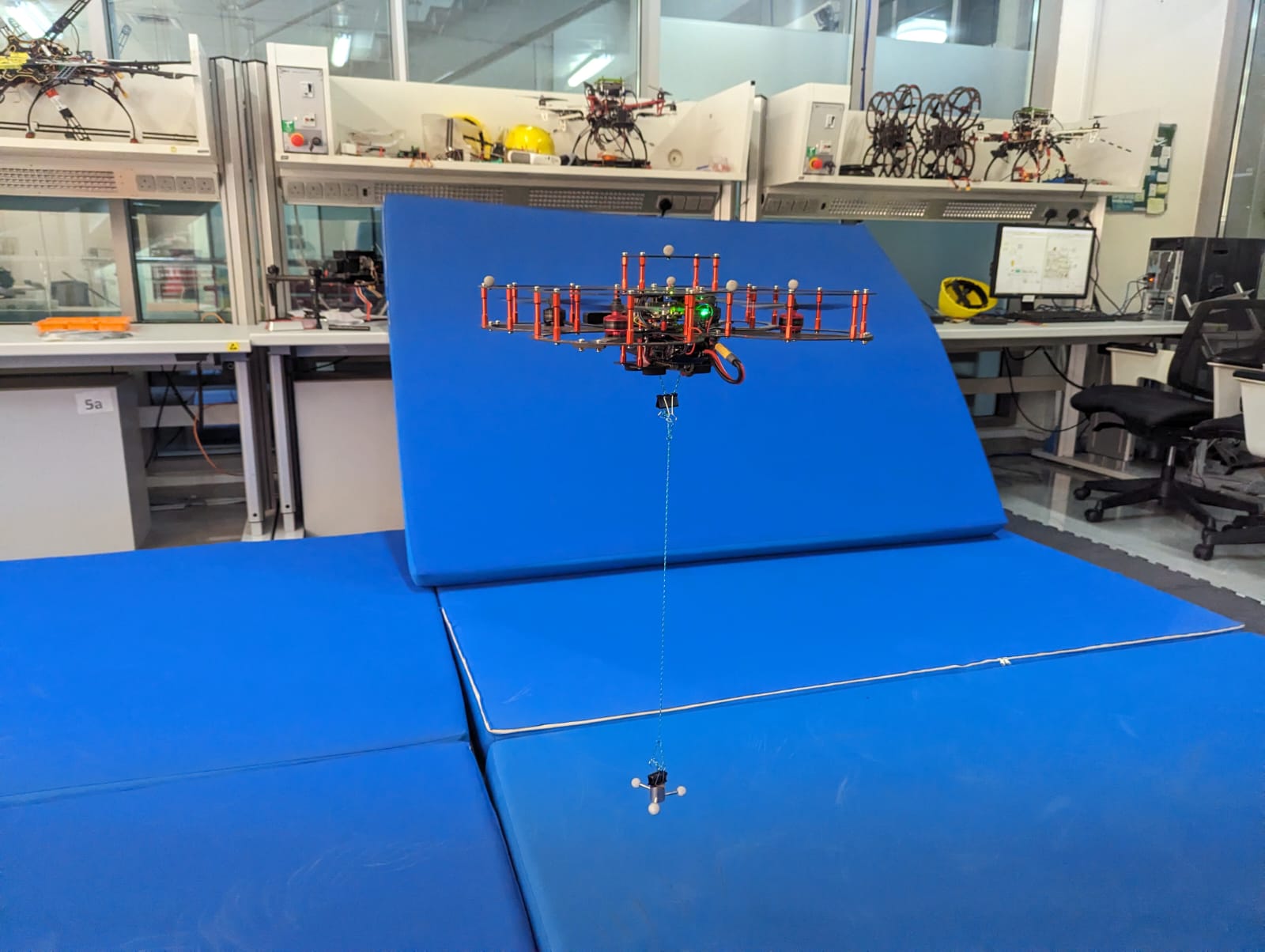}}
\caption{The experimental setup. A slung load is attached to a mid-size quadrotor through a cable. The cable length is changed between experiments to demonstrate the dual-scale homogeneity property.}
\label{setup}
\end{figure}

The position of the UAV is shown in Fig.~\ref{exp1} where the system with longer cable length has an identical amplitude profile to the shorter length which agrees with the simulation results. It also illustrates the homogeneity property that holds in the angle response as it was demonstrated in the theoretical proofs and simulations. Although the homogeneity property is not exact due to the additional dynamics of the UAV and its controllers, the transformations maintained the first peaks almost exactly at the same amplitude with an expansion in the time-scale. Hence, the method can be used effectively in applications where exact angle amplitude is required such as bringing loads to a desired elevated surface.
\begin{figure}[htbp]
\centerline{\includegraphics[width=9cm]{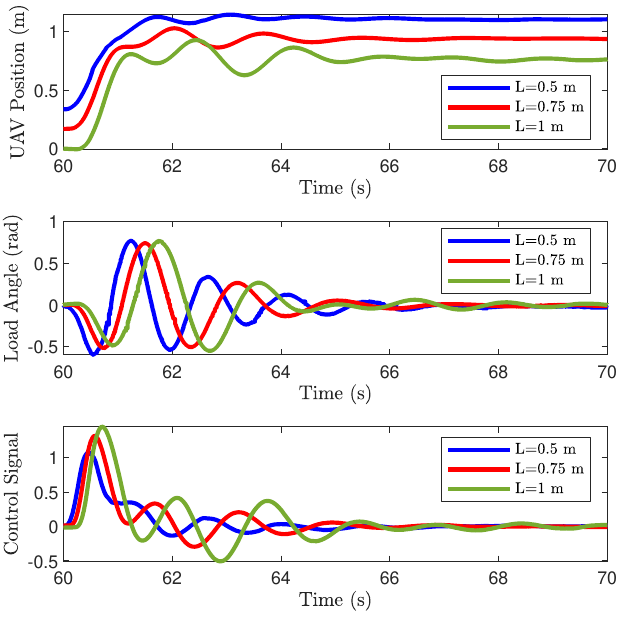}}
\caption{Experimental homogeneous responses of the UAV position, load angle, and control signal for three cable lengths.}
\label{exp1}
\end{figure}

The case in which the cable length is varied without applying homogeneity transformations is depicted in Fig.~\ref{exp2}. The feedback signals from the angle were disabled to observe a higher number of oscillations for illustration purposes. It can be seen that the oscillations have a very different spatial profile in this case.
\begin{figure}[htbp]
\centerline{\includegraphics[width=9cm]{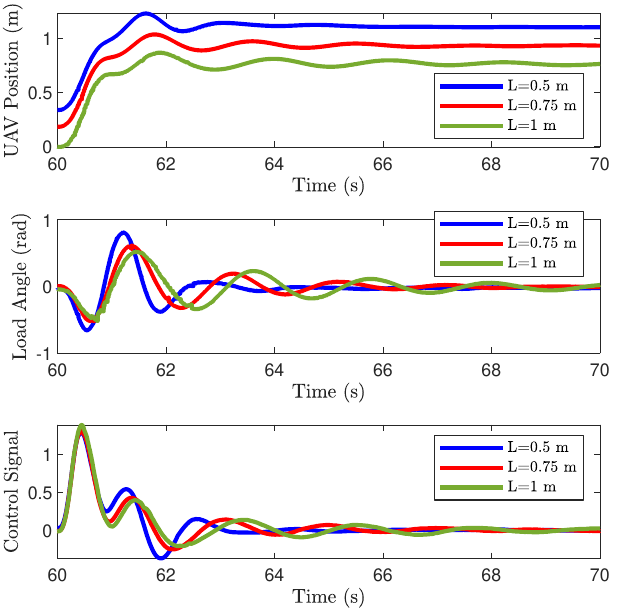}}
\caption{Experimental responses of the UAV position, load angle, and control signal for three cable lengths. No homogeneity transformations were used in this setting.}
\label{exp2}
\end{figure}

The impact of the presented results can be harnessed in applications where the load swing can be used during the pick up or place process. This is acheivable through the homogeneity transformations which allow the load to swing to a specific position regardless of the change in the cable length. The nonhomogeneous load spatial response in Fig.~\ref{spatial1} shows a high deviation in the load position at the first swing when the cable length is varied.
\begin{figure}[htbp]
\centerline{\includegraphics[width=9cm]{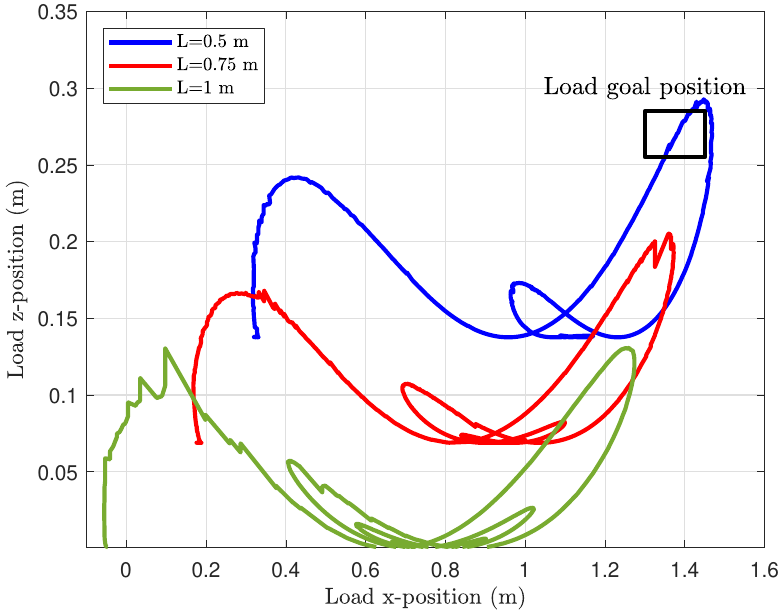}}
\caption{Experimental spatial responses of the load position for three cable lengths. No homogeneity transformations were used in this setting.}
\label{spatial1}
\end{figure}

The situation is different when homogeneous transformations are applied to the RL policy as depicted in the responses of Fig.~\ref{spatial2}. The load reaches the goal position at the first oscillation as guaranteed by the homogeneity property. This happens for all three cable lengths, in fact, for any cable length as long as the physical space is sufficient and motor saturation is not reached. Given that the oscillation angle is governed by the homogeneity property, the initial position of the quadrotor is adjusted in both the $x$ and $z$ positions in order for the load to reach its target despite the change in the cable length. The step size of the quadrotor horizontal position which results in the $45^{\circ}$ peak oscillation is $0.75$ m. This is a design choice based on the relative elevation of the goal position compared to the UAV, and it can be included in the policy training to have it as a reference input. The homogenized RL policy achieves a success rate of 96\% in bringing the load to its designated target with a 3D RMSE of $0.0253$ m based on 25 trials.
\begin{figure}[htbp]
\centerline{\includegraphics[width=9cm]{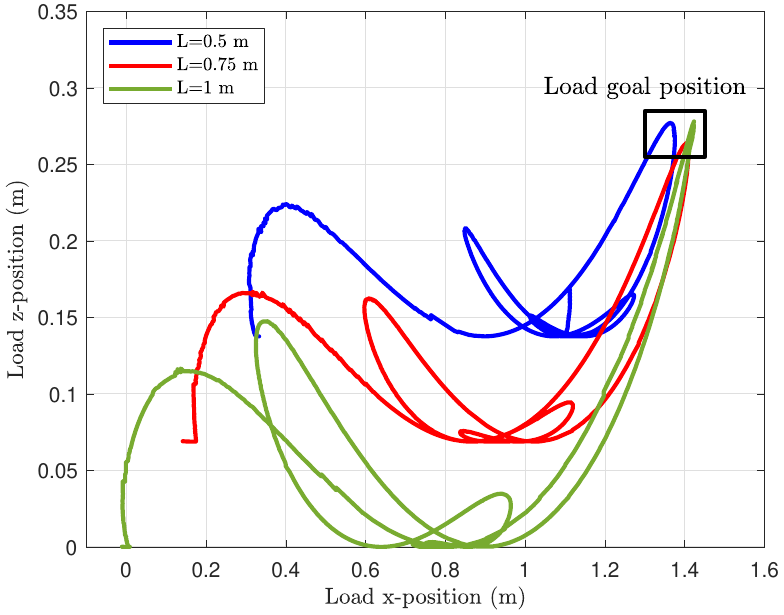}}
\caption{Experimental spatial responses of the load position for three cable lengths. Homogeneity transformations are applied.}
\label{spatial2}
\end{figure}

\section{Conclusions}
\label{conclusion}
This paper presents a new homogeneity property for dynamic systems for which spatio-temporal scaling holds. A control is developed for specific nonlinear systems that mathematically enforces the dual-scale homogeneity under any perturbation in the system's physical parameters. To demonstrate the proposed algorithm and its applicability to modern control techniques, the base control policy is generated by a deep deterministic policy gradient to swing up and stabilize an inverted pendulum for a single set of parameters. The modified control policy showed successful swing up and stabilization for different sets of mass and length parameters. The proposed algorithm in which classical control and reinforcement control are combined compares favorably to domain randomization and has a strong potential for the control of complex nonlinear systems. Extension of this work may include theoretical derivation of the effect and mitigation of time delay in the dual-scale homogeneity.

\section*{Acknowledgments}
This work was supported by the project fund RIG-2023-076 and Award RCI-2018-KUCARS, Khalifa University.

\bibliographystyle{IEEEtran}
\bibliography{bibli}

\begin{IEEEbiography}[{\includegraphics[width=1in,height=1.25in,clip,keepaspectratio]{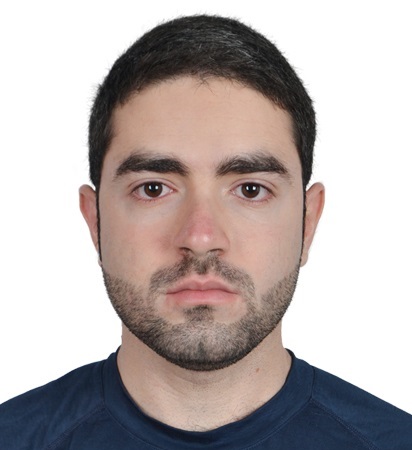}}]{Abdel Gafoor Haddad} earned his MSc degree from Khalifa University, Abu Dhabi, United Arab Emirates. He is currently pursuing his PhD degree in Engineering - Robotics concentration at Khalifa University. His research focuses on control systems and reinforcement learning for robotics applications.
\end{IEEEbiography}

\begin{IEEEbiography}[{\includegraphics[width=1in,height=1.25in,clip,keepaspectratio]{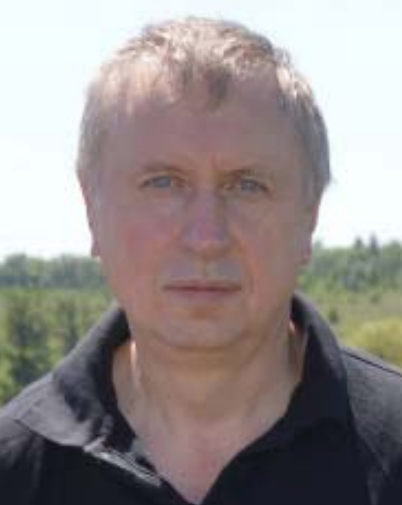}}]{Igor Boiko} (Senior Member, IEEE) received the M.Sc., Ph.D., and D.Sc. degrees from Tula State University, Tula, Russia, and Higher Attestation Commission, Russia, in 1984, 1990, and 2009, respectively. Currently, he is a Professor at Khalifa University, Abu Dhabi, UAE. His research interests include frequency-domain methods of analysis and design of nonlinear systems, discontinuous and sliding mode control systems, PID control, and process control theory and applications.
\end{IEEEbiography}

\begin{IEEEbiography}[{\includegraphics[width=1in,height=1.25in,clip,keepaspectratio]{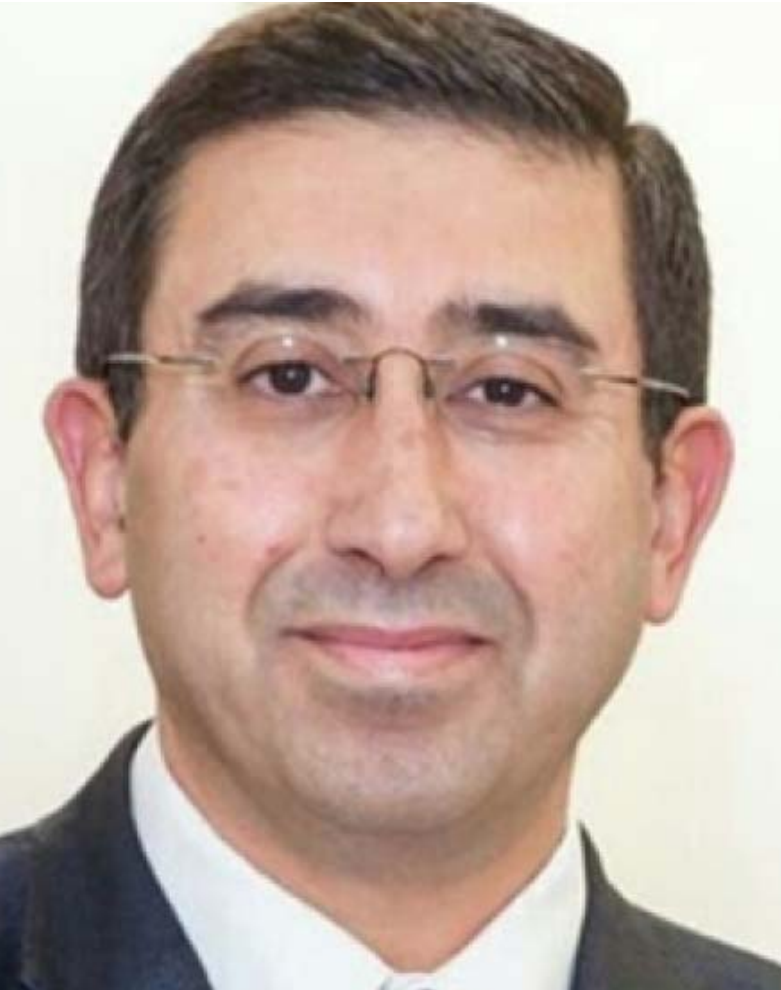}}]{Yahya Zweiri} (Member, IEEE) obtained his Ph.D. degree from King's College London. He is currently a Professor at the Department of Aerospace Engineering and the Director of the Advanced Research and Innovation Center, Khalifa University, United Arab Emirates. Over the past two decades, he has actively participated in defense and security research projects at institutions such as the Defense Science and Technology Laboratory, King's College London, and the King Abdullah II Design and Development Bureau in Jordan. Dr. Zweiri has a prolific publication record, with over 130 refereed journals and conference papers, as well as ten filed patents in the USA and UK. His primary research focus centers around robotic systems for challenging environments, with a specific emphasis on applied AI and neuromorphic vision systems
\end{IEEEbiography}

\vfill

\end{document}